\documentclass[aip,jmp,preprint]{revtex4-1}
\usepackage{amsmath,amssymb,amsmath,amsthm,amsopn}
\usepackage{mathtools}
\usepackage{fixltx2e}
\usepackage{units,xfrac}
\PassOptionsToPackage{normalem}{ulem}
\usepackage{ulem}
\usepackage[inline]{enumitem}
\usepackage{multirow}
\usepackage{nameref}
\usepackage{graphicx,tikz}
\usetikzlibrary{patterns}
\usepackage{natbib}

\theoremstyle{theorem}
\newtheorem{thm}{\protect\theoremname}[section]
\newtheorem*{thm*}{\protect\theoremname}
\newtheorem{lem}{\protect\lemname}[section]

\newtheorem{prop}[lem]{\protect\propname}

\newtheorem{ass}{\protect\assname}[section]
\theoremstyle{definition}

\theoremstyle{remark}
\newtheorem*{rem*}{\protect\remname}
\newtheorem*{rems*}{\protect\remsname}

\providecommand{\theoremname}{Theorem}
\providecommand{\lemname}{Lemma}
\providecommand{\corname}{Corollary}
\providecommand{\propname}{Proposition}
\providecommand{\conjname}{Conjecture}
\providecommand{\defnname}{Definition}
\providecommand{\assname}{Assumption}
\providecommand{\remname}{Remark}
\providecommand{\remsname}{Remarks}

\newcommand{\ipc}[2]{\left \langle #1 , \ #2 \right \rangle }
\newcommand{\Ev}[1]{\E \left( #1 \right)}  

\newcommand{\norm}[1]{\left\Vert#1\right\Vert}
\newcommand{\abs}[1]{\left\vert#1\right\vert}

\newcommand{\setb}[2]{\left \{ #1 \ \middle | \ #2 \right \} }
\newcommand{\com}[2]{\left[ #1 , #2 \right ]}
\newcommand{\bra}[1]{\left < #1 \right |}
\newcommand{\ket}[1]{\left | #1 \right >}
\newcommand{\dirac}[3]{\bra{#1} #2 \ket{#3}} 
\newcommand{\diracip}[2]{\left <#1 \middle | #2 \right >}

\newcommand{\bb}[1]{\mathbb{#1}}
\newcommand{\mc}[1]{\mathcal{#1}}
\newcommand{\wt}[1]{\widetilde{#1}}
\newcommand{\wh}[1]{\widehat{#1}}
\renewcommand{\vec}[1]{\mathbf{#1}}

\def\Z{\mathbb Z}

\def\R{\mathbb R}

\def\E{\mathbb E}

\def\e{\mathrm e}
\def\im{\mathrm i}
\def\di{\mathrm d}

\def\1{{\mathsf 1}}
\def\tem{\textemdash}

\def\ran{\operatorname{ran}} 
\def\Re{\operatorname{Re}} 

\usepackage[letterpaper]{geometry}
\usepackage[unicode=true, bookmarks=true,bookmarksnumbered=false,bookmarksopen=false, breaklinks=false,pdfborder={0 0 1},backref=false,colorlinks=false]{hyperref}
\hypersetup{pdfauthor={Jeffrey Schenker}}
\numberwithin{equation}{section}
\numberwithin{figure}{section}

\begin{document}
\title{Quantum Brownian motion induced by thermal noise in the presence of disorder} 
\author{J\"urg Fr\"ohlich}
\affiliation{Institute of Theoretical Physics, ETH Z\"urich, CH-8093 Zurich, Switzerland }
\email{juerg@phys.ethz.ch}
\author{Jeffrey Schenker}
\affiliation{Michigan State University Department of Mathematics, Wells Hall, 619 Red Cedar Road, East Lansing, MI 48824, USA }
\email{jeffrey@math.msu.edu}

\begin{abstract}
The motion of a quantum particle hopping on a simple cubic lattice under the influence of thermal noise and of a static random potential is expected to be diffusive, i.e., the particle is expected to exhibit \textit{`quantum Brownian motion'}, no matter how weak the thermal noise is.
This is shown to be true in a model where the dynamics of the particle is governed by a Lindblad equation for a one-particle density matrix. The generator appearing in this equation is the sum of two terms: a Liouvillian corresponding to a random 
Schr\"{o}dinger operator and a Lindbladian describing the effect of thermal noise in the kinetic limit.
Under suitable but rather general assumptions on the Lindbladian, the diffusion constant characterizing the asymptotics of the motion of the particle is proven to be strictly positive and finite. If the disorder in the random potential is so large that transport is completely suppressed in the limit where the thermal noise is turned off, then the diffusion constant tends  to zero proportional to the coupling of the particle to the heat bath.	
\end{abstract}

\maketitle

\section{Introduction}
In this paper, we study the propagation of a quantum particle through the lattice $\mathbb{Z}^d$ in the presence of a disordered potential landscape and under the influence of thermal noise. The dynamics of the particle is described by a Lindblad equation for the time evolution of its state, a one-particle density matrix. The Hamiltonian part in the total Lindblad generator is given by a random Schr\"odinger operator (i.e.,  an Anderson Hamiltonian), the dissipative part, which describes the dynamical effects of the thermal noise, by a Lindblad operator that couples the motion of the particle to the degrees of freedom --- afterwards ``traced out'' --- of a  heat bath in thermal equilibrium at a positive temperature $\beta^{-1}$. If the coupling of the particle to the heat bath is turned off, its dynamics is generated by a standard random 
Schr\"odinger operator. It is well known that strong disorder in the potential landscape may completely suppress coherent transport of the quantum particle --- the particle gets stuck near a well of the random potential. This is the phenomenon of \textit{Anderson localization}. It is natural to ask whether localization survives when the particle is coupled to the heat bath. The \textit{main result} established in this paper (see Theorem \ref{thm:main}, below) says that an arbitrarily small amount of thermal noise suffices to destroy Anderson localization. The particle then exhibits a diffusive motion (``quantum Brownian motion''). We will further show that, under the assumption of complete localization in the absence of thermal noise, the diffusion constant tends to zero, as the coupling to the heat bath is turned off. In contrast, in the absence of disorder, the diffusion constant diverges, as the coupling between the particle and the heat bath is turned off, because, in this limit, transport of the particle is governed by ballistic motion.

This paper is an adaptation to the present context of work of the second author on dynamical randomness\cite{Schenker2014} and makes use of similar mathematical techniques.

To be concrete, we consider a quantum particle hopping on the lattice $\Z^d$ whose pure states are given by wave functions, $\psi_t \in \ell^{2}(\mathbb{Z}^d)$, evolving according to the Schr\"odinger equation
$$\partial_t \ket{\psi_t} \ = \ -\im H_\omega \ket{\psi_t}$$
where $t$ denotes time.
In this equation, $H_\omega$ is an Anderson (random Schr\"odinger) Hamiltonian of the form
\begin{equation}\label{eq:andersonHam} H_\omega \ = \  u \sum_{|x-y|=1} \ket{x}\bra{y} +  \lambda \sum_{x} \omega(x) \ket{x}\bra{x}
\end{equation}
where $u$ is the hopping amplitude, $\omega$ is a random potential, with $\{ \omega(x) \}_{x\in \mathbb{Z}^d}$ independent identically distributed (i.i.d.) random variables, and $\lambda$ is a constant characterizing the strength of the disorder. It is often assumed that the distribution of the variables $\omega(x)$ has a bounded density. In this paper, we only need to assume that if the disorder is large all states are localized in the sense that the spectrum of $H_{\omega}$ is dense pure-point, and the corresponding eigenfunctions are localized. This is known if the distribution of the variables 
$\omega(x)$ is bounded. Assuming that the ratio $\nicefrac{u}{\lambda}$ is sufficiently small, there then exists a constant $\mu>0$ such that if $\psi_{t=0}$ has finite support, e.g., 
$\langle x\vert\psi_{t=0}\rangle =\delta_{x0}$, then
\begin{equation}\label{eq:AL}\Ev{\abs{\diracip{x}{\psi_t}}^2} \ \le \ C \e^{-\mu|x|}\end{equation}
where $\Ev{\cdot}$ denotes averaging with respect to the random potential $\omega$, and $C$ is a finite constant (depending on the choice of $\psi_{t=0}$); see \citet{Frohlich1983} and  \citet{Aizenman1993}. See also Ref.\ \onlinecite{Schenker2014a} for a recent survey with an estimate of the critical value of $\nicefrac{u}{\lambda}$. 

In concrete situations of physics, a quantum particle will usually interact, at least weakly, with degrees of freedom describing a surrounding medium, e.g., with the phonons corresponding to quantized vibrations of a crystal lattice. If these degrees of freedom are excited it is conceivable that, even at $large$ disorder, the particle exhibits \emph{dissipative} transport \tem \ diffusive motion, or motion with friction \tem \ due to its interactions with the medium. This is the phenomenon studied in this paper.

A physical theory of dissipative electron transport in conducting materials was developed a long time ago; see, e.g., \citet{Mott1968}.  More recently the mathematical formalism for the quantum Markov approximation in aperiodic and disordered media was studied by \citet{Spehner}, although quantum diffusion was not established in that work.  A simple one-particle model exhibiting quantum diffusion was studied in Ref.\ \onlinecite{DeRoeck2011}, without a quantum Markov approximation. A formalism for the study of dissipative transport in disordered semi-conductors,  within a quantum Markov approximation but incorporating effects of the Pauli principle on a gas of non-interacting electrons at positive density, was developed by \citet{Androulakis}. The purpose of the study undertaken in Ref.\ \onlinecite{Androulakis} was to gain some understanding of the role of thermal noise in the quantum Hall effect. In comparison, the goals of the present paper are more modest. Assuming that the density of particles (moving, e.g., in a conduction band of a semi-conductor) is so small that the one-particle approximation can be justified, we propose to analyze the interplay between randomness and thermal noise in the transport of a single quantum particle. We plan to analyze systems at positive density in future work.

Next, we introduce some notation and describe the model studied in this paper.
The state of an open system consisting of a single quantum particle that interacts with a dynamical environment (medium) is typically a mixed state, which is described by a one-particle density matrix, $\rho_t$ (with $t$ denoting time). This one-particle description is obtained by tracing out all degrees of freedom of the environment. 
In the kinetic (Markovian) limit, the evolution of $\rho_t$ is described by a a Lindblad equation of the form 
\begin{equation} 
\partial_t \rho_t  \ = \ -\im \com{H_\omega}{\rho_t} + g \mc{L}\rho_t,
\label{eq:Lindblad}
\end{equation}
where $H_\omega$ is a Hamiltonian, in the following the one introduced in eq.\ \eqref{eq:andersonHam}, g is a constant that measures the strength of the thermal noise, and the operator $\mc{L}$ is a \textit{Lindblad generator}, which we specify below. 

In our context, a one-particle density matrix, $\rho$, is a positive, trace-class operator on 
$\ell^2(\Z^d)$. 
In the $x$-space representation, the operator $\rho$ has matrix elements $\dirac{x}{\rho}{y}$, with
$(x,y) \in \mathbb{Z}^d \times \mathbb{Z}^d$. It is convenient to introduce the variables
\begin{equation}
X= x+y, \qquad \xi= x-y\ .
\label{newvar}
\end{equation}
Then $X\in \mathbb{Z}^d$ and $\xi \in \mathbb{Z}^d$, with 
$$ X \pm \xi \in 2\mathbb{Z}^d\ .$$
From now on, we will write $\rho(X,\xi)$ for the matrix element  $\dirac{x}{\rho}{y}$ of $\rho$, where it is understood that $(X,\xi)$ and $(x,y)$ are related by eq.\ \eqref{newvar}.  
Consistent with eq.\ \eqref{newvar} we adopt the convention that $\rho(X,\xi)=0$ if $X+\xi$ or $X-\xi$ do not belong to $2 \Z^d$. 

To define the action of the Lindblad operator $\mathcal{L}$ on the density matrix $\rho$ we introduce two operators, $G$ (for ``gain'') and $L$ (for ``loss''), defined as follows:
\begin{equation}\label{eq:gain}
(G\rho)(X,\xi) \ = \ \sum_{\eta\in \Z^d} r(\xi,\eta) \rho(X,\eta)
\end{equation}
and
\begin{equation}\label{eq:loss}
(L\rho)(X,\xi) \ = \ \sum_{\eta \in \Z^d} r(0,\eta-\xi) \rho(X,\eta)	
\end{equation}
where the ``gain kernel'' $r:\Z^d\times \Z^d \rightarrow \bb{C}$  satisfies certain properties specified below.  More generally, the operators $G$ and $L$ might also act on the variable $X\in\mathbb{Z}^d$ by convolution (which is imposed by the requirement of translation invariance), but this possibility will not be considered in the present paper. 

The Lindblad operator $\mathcal{L}$ is expressed in terms of $G$ and $L$ by
\begin{equation}\label{eq:Lspaceform}
(\mc{L}\rho)(X,\xi) \ = \ (G\rho)(X,\xi) - (L\rho)(X,\xi)
\end{equation}
so that the Lindblad equation \eqref{eq:Lindblad} for $\rho_t$, incorporating hopping, disorder and dissipation, takes the form  
\begin{equation}\label{eq:SErho}
\begin{aligned}
 \partial_t \rho_t(X,\xi) \ =  \ & -\im u \sum_{|e|=1} \left [ \rho_t(X+e,\xi+e) - \rho_t(X+e,\xi-e) \right ] 	 \\ 
& \ - \ \im \lambda \left [ \omega\left ( \frac{X+\xi}{2} \right ) - \omega \left ( \frac{X-\xi}{2} \right ) \right ] \rho(X,\xi)\\
& \ + \ g \sum_{\eta\in \mathbb{Z}^d} \left [ r(\xi,\eta) - r(0,\eta-\xi) \right ] \rho(X,\eta) \ .
\end{aligned}
\end{equation}  

The gain kernel $r$ is usually assumed to have certain physically natural properties, such as ``detailed balance'' for the energy.  These properties are discussed in \S\ref{sec:properties} below. Here we only note that, for our analysis to proceed, $r$ must satisfy the following
\begin{ass}\label{ass:r} The gain kernel $r:\Z^d\times \Z^d \rightarrow \bb{C}$ satisfies: 
\begin{enumerate}
\item $r(\xi,\eta)=0$ unless $\xi+\eta\in 2\Z^d$;
\item $r$ the Fourier transform of a non-negative measure $\mu$,
\begin{equation}\label{eq:r} r(\xi,\eta) \ = \ \int_{\bb{T}^d\times \bb{T}^d} \e^{\im \xi \cdot \vec{p} - \im \eta \cdot\vec{q}} \, \mu(\di \vec{p},\di \vec{q}),
\end{equation}
where $\mathbb{T}^d$ denotes the $d$-torus $\mathbb{T}^d=[0,2\pi)^d$;
\item $r(\xi,0)-r(0,-\xi)=0$ for each $\xi\in \Z^d$; and
\item for any $\phi \in \ell^2(\Z^d)$ we have that
\begin{equation}\label{eq:exp-bound} \mathrm{Re} \sum_{\xi\in \Z^d}\sum_{\eta\in \Z^d} \left [ r(0,\eta-\xi)- r(\xi,\eta) \right ] \overline{\phi(\xi)}\phi(\eta) \ \ge \ c \sum_{\xi\neq 0} \abs{\phi(\xi)}^2 
\end{equation}
\item $r$ is the kernel of a bounded operator on $\ell^2(\Z^d)$, i.e.,
$$ \sum_{\xi\in \Z^d} \abs{\sum_{\eta\in \Z^d}r(\xi,\eta) \phi(\eta)}^2 \ \le \ C \sum_{\xi\in \Z^d}\abs{\phi(\xi)}^2 $$
for some $C<\infty$ and any $\phi \in \ell^2(\Z^d)$;
\end{enumerate}
\end{ass}

\begin{rem*} 
1) Item 1 is required for the gain operator $(G\rho)(X,\xi)$ to satisfy the constraint of vanishing whenever 
$X\pm \xi \not \in 2\bb{Z}^d$.  
Items 1 through 5 are discussed in  \S\ref{sec:properties} below.\\
2) Since $r$ is the kernel of a bounded operator on $\ell^2$, it follows that $\mc{L}$ is a bounded map (from $\ell^2(\Z^d) \rightarrow \ell^2(\Z^d)$), on each fiber over $X\in \Z^d$.  Without loss, we scale $r$ and $g$ such that $\mc{L}$ has norm $\le 1$, i.e.,
\begin{equation}\label{eq:normalization} \sum_{\xi} \abs{\mc{L}\rho(X,\xi)}^2 \ \le \ \sum_{\xi} \abs{\rho(X,\xi)}^2.	
\end{equation}
\end{rem*}

For simplicity, we further assume that the Lindblad generator $\mathcal{L}$ respects certain lattice symmetries:
\begin{ass}\label{ass:symmetry}
\begin{enumerate}

\item $r(\xi,\eta)$ is invariant under inversion of coordinates.  That is, for each $i=1,\ldots,d$, 
\begin{equation}\label{eq:invert}
r(R_i \xi, R_i \eta) \ = \ r(\xi,\eta)	
\end{equation}
where 
$$R_i (x_1,\ldots,x_i,\ldots, x_d) \ = \ (x_1,\ldots,-x_i, \ldots, x_d)$$
\item $r(\xi,\eta)$ is invariant under permutation of coordinates.  That is, for any permutation $\sigma$ of $\{1,\ldots,d\}$,
\begin{equation}
r(T_\sigma \xi, T_\sigma \eta) \ = \ r(\xi,\eta)
\label{eq:permute}
\end{equation} 
where 
$$T_\sigma(x_1,\ldots,x_d) \ = \ (x_{\sigma(1)},\ldots,x_{\sigma(d)}). $$ 
\end{enumerate}
\end{ass}
From Assumption \ref{ass:symmetry} and eq.\ \eqref{eq:SErho} it follows  that the processes $(t,X,\xi) \mapsto \rho_t(T_\sigma X,T_\sigma \xi)$, $(t,X,\xi)\mapsto \rho_t(R_iX,R_i\xi)$ and $(t,X,\xi) \mapsto \rho_t(X,\xi)$ all have the same distribution.

The following theorems are the \textit{main results} proved in this paper.

\begin{thm}\label{thm:main}
Let $\rho_{t}$ be a solution of eq.\ \eqref{eq:SErho} with initial condition $\rho_{t=0} = \ket{0}\bra{0}$, and suppose that $r$ satisfies Assumptions \ref{ass:r} and \ref{ass:symmetry}. 
If the particle interacts with the heat bath, i.e., $g >0$, then the diffusion constant
\begin{equation}\label{eq:D} D_{i,j} \ := \ \lim_{t \rightarrow \infty } \frac{1}{t} \sum_{x\in \Z^d} x_i x_j \Ev{\dirac{x}{\rho_t}{x}} \ = \  \lim_{t \rightarrow \infty } \frac{1}{4t} \sum_{X\in 2 \Z^d } X_i X_j \Ev{ \rho_t(X,0)}	
\end{equation}
exists and satisfies
\begin{equation}\label{eq:D-identity}D_{i,j} \ = \ D  \delta_{i,j}	
\end{equation}
with $0 < D = d^{-1} \displaystyle{ \sum_{i=1}^d} D_{i,i} <\infty$.
\end{thm}

Let us now consider the behavior of the diffusion constant in the limit where the coupling with the heat bath tends to zero, $g\rightarrow 0$.  We are interested in the behavior of the diffusion constant $D$ as a function of $g$, so we will henceforth write $D(g)$. 
In particular, we would like to estimate the size of $D(g)$ if the hopping $u$ and the disorder strength are such that the strong localization result described in eq.\ \eqref{eq:AL} holds. 
In fact, we will not use the full strength of eq.\ \eqref{eq:AL}.
Instead, what is required below is simply a uniform bound on the second moment of the position, viz.
\begin{equation} \ell^2 \ := \ \sup_{t}  \sum_{x\in \Z^d} \frac{|x|^2}{d} \Ev{ \abs{\dirac{x}{\e^{-\im H_\omega t} }{0}}^2} < \infty \ .\label{eq:localization}
\end{equation}
\begin{thm}\label{thm:dissloc} Suppose eq.\ \eqref{eq:localization} holds.  Then 
\begin{equation}\label{eq:dissloc} D(g) \ = \ \Delta   g + o(g) \ ,	
\end{equation}
with $0< \Delta  < (1 + \nicefrac{1}{c}) \ell^2$.  (Here $c$ is the constant appearing in item 5 of assumption \ref{ass:r}.)
\end{thm}

By way of contrast, in the absence of disorder ($\lambda=0$ in eq.\ \eqref{eq:andersonHam}) we have the following elementary result.
\begin{thm}\label{thm:dissbal}
If $\lambda =0$ then
\begin{equation}
\label{eq:dissbal}	 D(g) \ = \ C \frac{u^2}{g}
\end{equation}
with $0 < C < \nicefrac{4}{c}$.
\end{thm}
We observe that the diffusion constant diverges, as $g\rightarrow 0$; the reason being that, without disorder and without dissipation, the motion of the particle is \textit{ballistic}.

\subsection{Properties of the gain kernel}\label{sec:properties}
Let $\di \vec{q}$ denote \emph{normalized} Haar measure on the $d$-torus $\mathbb{T}^d$,  $\int_{\mathbb{T}^d} \di \vec{q} = 1$.  The measure $\mu$ appearing in eq.\ \eqref{eq:r} need not be absolutely continuous with respect to the product measure $\di \vec{p}\times \di \vec{q}$ on $\mathbb{T}^d\times \bb{T}^d$.  
Nonetheless, we will write (using distributional notation)
$$ \mu(\di \vec{p},\di \vec{q}) \ = \ \wh{r}(\vec{p},\vec{q}) \di \vec{p} \di \vec{q}$$
where $\wh{r}$ is a (possibly singular) positive distribution on $\bb{T}^d\times \bb{T}^d$ \tem \ see eq.\ \eqref{eq:r-example} below. 
In terms of $\wh{r}$, the gain and loss operators have the following expressions 
\begin{equation}\label{W gain term}
(G \rho^W)(X,\vec{p}) \ := \ \int_{\mathbb{T}^d} \rho^W(X,\vec{q})	\wh{r}(\vec{p},\vec{q}) \di \vec{q}
\end{equation}
and 
\begin{equation}
\label{W loss term}
(L \rho^W)(X, \vec{p})\ := \ \left [ \int_{\mathbb{T}^d} \wh{r}(\vec{q},\vec{p}) \di \vec{q}
 \right ]  \rho^W(X,\vec{p})
\end{equation}
Here $\rho^W$ is the \emph{Wigner transform} of $\rho$, which is the Fourier transform of $\rho(X,\xi)$ in the $\xi$ variable:
\begin{equation}\label{eq:Wigner}
\rho^W (X,\vec{p}) \ := \ \sum_{\xi \in \Z^d} \e^{\im \vec{p}\cdot \xi} \rho(X,\xi)
\end{equation} 
The variable $\vec{p}$ is the \emph{quasi-momentum} of the particle; it ranges over the $d$-torus $\mathbb{T}^d$, which is the  Brillouin zone corresponding to $\mathbb{Z}^d$.  
Although it need not be positive, the function $\rho^W(X,\vec{p})$ is interpreted as a density for the distribution of the particle position and quasi-momentum in phase space $\mathbb{Z}^d\times \mathbb{T}^d$.  

When acting on the Wigner transform the Lindblad operator has the form
\begin{multline}
\label{W Lindblad}
(\mc{L} \rho^W)(X,\vec{p})  \ := \  (G\rho^W)(X,\vec{p}) - (L\rho^W)(X,\vec{p})\\ = \
 \int_{\mathbb{T}^d} \left [ \wh{r}(\vec{p}, \vec{q}) \rho^W(X,\vec{q})-\wh{r}(\vec{q}, \vec{p})\rho^W(X,\vec{p}) \right ]  \di \vec{q}
\end{multline}
Thus, in the fiber over $X$, $\mc{L}$ acts as the generator of a hopping process in the quasi-momentum:  $\wh{r}(\vec{p},\vec{q}) \di \vec{q}$ is the rate at which the quasi-momentum of the particle jumps from $\vec{q}$ to $\vec{p}$. 
To be able to interpret $\mc{L}$ as the generator of a stochastic process, we need $\wh{r}$ to be a positive measure. Thus we take $r$ to be of positive type \tem \ item 2 of Assumption \ref{ass:r}. 

Item 3 of Assumption \ref{ass:r} is equivalent to
\begin{equation}\label{eq:fucondition}\int_{\mathbb{T}^d} \wh{r}(\vec{p},\vec{q}) \di \vec{q} \ = \ \int_{\mathbb{T}^d} \wh{r}(\vec{q},\vec{p})\di \vec{q}	\ .
\end{equation}
Thus the generator of the jump process may just as well be written as  
\begin{equation}\label{eq:altL}
(\mc{L} f)(\vec{p})  \  = \
 \int_{\mathbb{T}^d} \di \vec{q} \, \wh{r}(\vec{p}, \vec{q}) \left [  f(\vec{q})-f(\vec{p}) \right ]
\end{equation}
for $f\in L^2(\mathbb{T}^d)$.
It follows that $\mathcal{L} \mathbf{1}=0$ where $\mathbf{1}$ is the identity function $\mathbf{1}(\vec{p})=1$ on the torus, i.e., Haar measure $\di \vec{p}$ is an invariant measure for the momentum jump process.  
In fact, by item 4 of Assumption \ref{ass:r}, Haar measure is the \emph{unique} invariant probability measure. 
Indeed, eq.\ \eqref{eq:exp-bound} is readily seen to be equivalent to the following estimate:
\begin{multline} 
- 2 \Re \ipc{f}{\mathcal{L}f}_{L^2(\mathbb{T}^d)} \ = \ 
\int_{\mathbb{T}^d\times \mathbb{T}^d}  \wh{r}(\vec{p}, \vec{q})\di \vec{p}\di \vec{q} \abs{f(\vec{p}) - f(\vec{q})}^2 \\
\ge \ c \int_{\mathbb{T}^d\times \mathbb{T}^d} \di \vec{p}\di \vec{q} \abs{f(\vec{p}) -f(\vec{q})}^2
\label{eq:condition}
\end{multline}
for every $f\in C(\mathbb{T}^d\times \mathbb{T}^d)$.  
Eq.\ \eqref{eq:condition} is a ``spectral gap'' condition for the generator of the jump process in quasi-momentum space. 
Assuming that this condition holds amounts to requiring exponentially fast mixing of the jump process.

Item 3 of Assumption \ref{ass:r} implies that the operators 
$$Gf(\xi) \ = \ \sum_{\eta}r(\xi,\eta) f(\eta) \quad \text{and} \quad Lf(\xi) \ = \ \sum_{\xi}r(0,\eta-\xi) f(\eta)$$
are bounded on $\ell^2(\Z^d)$.
A natural condition on the jump process which gives the bound is to require the total rate of jumping out of any fixed quasi-momentum $\vec{p}$ to be uniformly bounded, i.e.,
\begin{equation}\label{eq:boundedjumping}
	M \ := \ \sup_{\vec{p}} \int_{\bb{T}^d}  \wh{r}(\vec{q},\vec{p})\di \vec{q} \ < \ \infty \ .
\end{equation}
\begin{lem} If the measure $\wh{r}(\vec{p},\vec{q})\di \vec{p} \di \vec{q}$ satisfies eqs.\ \eqref{eq:fucondition} and \eqref{eq:boundedjumping} then $G$ and $L$ are bounded operators on $\ell^2(\Z^d).$
\end{lem}
\begin{rem*} That is, eqs.\ \eqref{eq:fucondition} and \eqref{eq:boundedjumping} are sufficient for item 3 of Assumption \ref{ass:r}.	
\end{rem*}
\begin{proof} 
Taking Fourier transforms we find that
$$\widehat{G\phi}(\vec{p}) \ = \  \int_{\bb{T}^d} 	\wh{r}(\vec{p},\vec{q}) \wh{\phi}(\vec{q}) \di \vec{q}\quad \text{and} \quad \widehat{L\phi}(\vec{p}) \ = \   \int_{\bb{T}^d} \wh{r}(\vec{q},\vec{p})\di \vec{q} \text{   }
 \wh{\phi}(\vec{p})$$
where $\wh{\phi}(\vec{p}) = \sum_{\xi} \e^{\im \xi\cdot \vec{p}} \phi(\xi)$ is the Fourier transform of $\phi$.  By Plancherel's Theorem,
$$\norm{L\phi}_{2} \ \le \ M \norm{\phi}_2 \ ,$$
with $M$ as in eq.\ \eqref{eq:boundedjumping}.
To bound $G\phi$, we note that
$$ \abs{\wh{G\phi}(\vec{p})}^2 \ \le \ M \int_{\bb{T}^d} \abs{\wh{\phi}(\vec{q})}^2\wh{r}(\vec{p},\vec{q})\di \vec{q}$$
by the Cauchy-Schwarz inequality, since by eq.\ \eqref{eq:fucondition} 
$$M \ = \ \sup_{\vec{p}\in \bb{T}^d} \int_{\bb{T}^d} \wh{r}(\vec{p},\vec{q})\di \vec{q}.$$ 
Integrating over $\vec{p}$ and applying Plancherel's Theorem again, we find that
$$ \norm{G\phi }_2 \ \le \ M \norm{\phi}_2\ . \qedhere$$
\end{proof}

The results we have arrived at, so far, are summarized in the following
\begin{prop}
Let $\wh{r}(\vec{p},\vec{q})\di\vec{p}\di\vec{q}$ be a positive measure on $\bb{T}^d\times\bb{T}^d$ that is periodic with period $\pi$ under joint translations of $\vec{p}$ and $\vec{q}$, i.e.,
\begin{equation}\label{eq:periodic}
\wh{r}(\vec{p}+\pi \vec{n},\vec{q}+\pi \vec{n}) \ = \ \wh{r}(\vec{p},\vec{q})	
\end{equation}
 If $r$ is defined by eq.\ \eqref{eq:r} and eqs.\ \eqref{eq:fucondition}, 
 \eqref{eq:condition} and \eqref{eq:boundedjumping} hold for $\widehat{r}$, then $r$ satisfies Assumption \ref{ass:r}.
\end{prop}
Eq.\ \eqref{eq:periodic} is required for item 1 of Assumption \ref{ass:r}, namely that $r(\xi,\eta)=0$ if $\xi+\eta\not \in 2\Z^d$, to hold.  
In this regard, it is useful to note that
$$\rho^W(X,\vec{p}+\vec{n}\pi) \ = \ \sum_{\xi \in \Z^d} \e^{\im (\vec{p}+\pi \vec{n})\cdot \xi} \rho(X,\xi) \ = \ \e^{\im \pi \vec{n} \cdot X} \rho^W(X,\vec{p})$$
for $\vec{n}\in \mathbb{Z}^d$, since in the sum over $\xi$ we have  $X- \xi \in 2 \mathbb{Z}^d$. 
Thus, $\rho^W$ is either periodic or anti-periodic with period $\pi$ in each component of the quasi-momentum, depending on whether the corresponding component of $X$ is even or odd.  Eq.\ \eqref{eq:periodic} guarantees that $G\rho^W$ and $L \rho^{W}$ also have this property.

Lindblad operators are often assumed to  satisfy a ``detailed balance condition'' at the temperature, $\beta^{-1}$, of the heat bath. 
For a \textit{free} particle detailed balance can be formulated as follows. 
Let $\varepsilon(\vec{p})$  denote the energy of a freely moving particle with quasi-momentum $\vec{p}$ \tem \ i.e., 
$$\varepsilon(\vec{p}) = 2u \left (d -\sum_{i=1}^{d} \cos(p_i) \right )$$
for the isotropic nearest neighbor hopping  considered here. 
Detailed balance at temperature  $\beta^{-1}$ is the condition that 
\begin{equation}
\label{detailed balance}
\wh{r}(\vec{p},\vec{q})  = \e^{\beta(\varepsilon(\vec{p})-\varepsilon(\vec{q}))}\wh{r}(\vec{q},\vec{p}) \ .
\end{equation}
For example, for a Lindblad generator describing the interaction of the particle with a thermal reservoir of non-interacting bosons with an even dispersion law $\omega(\vec{k})=\omega(-\vec{k})$ at temperature 
$\beta^{-1}$, one has that
\begin{multline} \label{eq:r-example}\wh{r}(\vec{p},\vec{q}) \ = \ F(\vec{p}-\vec{q}) \Bigg [ \frac{1}{1- \e^{-\beta \omega(\vec{p}-\vec{q})}} \delta(\varepsilon(\vec{p})-\varepsilon(\vec{q})-\omega(\vec{p}-\vec{q}))  \\
+ \frac{\e^{-\beta \omega(\vec{p}-\vec{q})} }{1- \e^{-\beta \omega(\vec{p}-\vec{q})}}\delta(\varepsilon(\vec{p})-\varepsilon(\vec{q})+\omega(\vec{p}-\vec{q})) \Bigg ]
\end{multline}
where $\delta(\cdot)$ denotes a Dirac delta function and the ``form factor'' $F$ is a non-negative, even function. 

In the present context detailed balance, as in eq.\ \eqref{detailed balance}, is only consistent with item 4 of Assumption \ref{ass:r} if $\beta=0$ (infinite temperature). 
At $\beta=0$, detailed balance is just symmetry of the gain kernel, $\wh{r}(\vec{p},\vec{q})=\wh{r}(\vec{q},\vec{p})$.  
However, symmetry is stronger than eq.\ \eqref{eq:fucondition} and detailed balance and symmetry, as such,  play no role in our analysis.
An explicit, and symmetric, example of a gain kernel satisfying our requirements is given by a suitable limit of eq.\ \eqref{eq:r-example} as $\beta \rightarrow 0$, i.e.,
\begin{equation} \label{eq:r-example2}\wh{r}(\vec{p},\vec{q}) \ = \ F(\vec{p}-\vec{q}) \Bigg [  \delta(\varepsilon(\vec{p})-\varepsilon(\vec{q})-\omega(\vec{p}-\vec{q}))  \
+ \ \delta(\varepsilon(\vec{p})-\varepsilon(\vec{q})+\omega(\vec{p}-\vec{q})) \Bigg ] \ .
\end{equation}
For a quantum particle whose dynamics (in the absence of thermal noise) is governed by a random 
Schr\"{o}dinger operator of the kind studied in this paper, detailed balance is a somewhat awkward condition. If the disorder is large it is natural to neglect the kinetic energy term of the particle in \eqref{detailed balance}, i.e., to set $\varepsilon \equiv 0$. In this case, our Assumption \ref{ass:r} is acceptable. A (not very natural) example of a gain kernel to which our analysis would apply, with $\beta>0$, can be found in Ref.\ \onlinecite{DeRoeck2011}. But these matters ought to be studied more thoroughly.

\section{The Diffusion Constant for Lindblad Dynamics with Disorder}
In this section we prove Theorem \ref{thm:main}. To begin with, we note that by Assumption \ref{ass:symmetry}
$$\sum_{X\in \Z^d } X_i X_j \Ev{ \rho_t(X,0)}  \ = \ 0\ , \quad i \neq j,$$
for all $t$. 
Indeed, let 
$$\wt{\rho}_t(X,\xi) \ = \ \rho_t(R_iX,R_i\xi) $$
where $R_i$ denotes  inversion of the $i^{\mathrm{th}}$ coordinate.  Because  $\wt{\rho_t}$ and $\rho_t$ have the same distribution, one has that
$$\sum_{X} X_i X_j \Ev{\wt{\rho}_t(X,0) } \ = \ \sum_{X} X_i X_j \Ev{\rho_t(X,0)}\ .$$
However, the definition of $\widetilde{\rho}_{t}$ implies that
$$\sum_{X} X_i X_j \Ev{\wt{\rho}_t(X,0) } \ = \ - \sum_{X} X_i X_j \Ev{\rho_t(X,0)}$$
if $i\neq j$.
Likewise, by permutation symmetry, all diagonal matrix elements are the same
$$\sum_{X} X_i^2 \Ev{\rho_t(X,0)} \ = \ \sum_{X} X_j^2 \Ev{\rho_t(X,0)}$$
$i,j=1,\ldots,d$. Thus the off-diagonal elements of the diffusion matrix $D_{i,j}$ defined by eq.\ \eqref{eq:D} vanish, and the diagonal elements are all equal, provided they are well defined. Thus, to analyze the diagonal elements of the diffusion matrix, it suffices to consider $D_{1,1}$.

Note that $\rho_t$ is a random variable depending on the disorder configuration $\omega\in \Omega$.  In what follows, we will emphasize this by writing $\rho_t(X,\xi,\omega)$. 
The initial condition is the density matrix $\rho_0$ given by
$$\rho_0(X,\xi,\omega) \ = \ \delta_0(X) \delta_0(\xi) {\bf{1}}_{\omega} \ .$$ 
We introduce the Hilbert space 
$$\mc{H} \ = \ \setb{\Psi \in L^2(\Z^d\times \Z^d\times \Omega)}{\Psi(X,\xi,\omega)=0\text{ if } X\pm \xi\not \in 2 \Z^d } \ .$$

Finite-group-velocity estimates for the propagation of a quantum particle on the lattice show that 
\begin{equation}\sum_{X,\xi} \e^{m|X|} \abs{\rho_t(X,\xi,\omega)}^2 \le  \e^{C_m t} \ .
\label{eq:FGV}
\end{equation}
Here $m$ and $C_m$ are positive, finite constants; see Lemma \ \ref{lem:FGVE}. 
In particular $\rho_t\in \mc{H}$, for each $t\ge 0$.  
 
Translation of the system by a lattice vector $a$ is given by a unitary map on $\mc{H}$ defined by
$$S_a \Psi(X,\xi,\omega) \ = \ \Psi(X-2a,\xi,\tau_{-a} \omega)$$
where $\tau_a \omega(x) \ = \ \omega(x-a)$.   
We define the following generalized Fourier transform on $\mc{H}$:
$$\mc{F}\Psi (\xi,\omega,\vec{k}) \ := \ \sum_{a} \e^{-\im \vec{k}\cdot a} S_{a}\Psi(\xi,\xi,\omega) \ = 
\ \sum_{a}  \e^{\im \vec{k}\cdot a} \Psi(\xi+2a,\xi,\tau_{a}\omega) \ .$$
This Fourier transform leads to a direct-integral decomposition of $\mc{H}$, with fibers isomorphic to 
$\wh{\mc{H}} :=  L^2(\Z^d\times \Omega)$.  Given $\Psi\in \mc{H}$, we let $\wh{\Psi}_{\vec{k}}(\xi,\omega) \ := \ \mc{F}\Psi (\xi,\omega,\vec{k})$. We then have that $\wh{\Psi}_{\vec{k}}\in \wh{\mc{H}}$, for almost every 
$\vec{k}$, and 
$$\norm{\Psi}_{\mc{H}}^2 \ = \ \int_{\bb{T}^d} \norm{\wh{\Psi}_{\vec{k}}}^2_{\wh{\mathcal{H}}} \di \vec{k}\ .$$
 
\begin{lem}\label{lem:main} 
Suppose that $\rho_t(X,\xi,\omega)$ solves eq.\ \eqref{eq:SErho}. Then
\begin{equation}
\label{Lindblad in fibres}
\partial_t \wh{\rho}_{\vec{k};t} \ = \  \left ( - i \mc{A}_\vec{k} + g \wh{\mc{L}}_{\vec{k}} \right )\wh{\rho}_{\vec{k};t}
\end{equation}
where $\mc{A}_{\vec{k}} \ = \ u \wh{T}_{\vec{k}}  + \lambda \wh{V}$, with
\begin{enumerate}
\item $\displaystyle\wh{T}_{\vec{k}} f(\xi,\omega) \ = \ \sum_{|e|=1} \left [ f(\xi+e,\omega) - \e^{-\im \vec{k}\cdot e} f(\xi+e,\tau_{e}\omega) \right ] ,$
\item $\displaystyle \wh{V} f(\xi,\omega) \ = \ \left (\omega(\xi) - \omega(0) \right ) f(\xi,\omega)$, \text{    }
and
\item $\displaystyle \wh{\mc{L}}_{\vec{k}}  f(\xi,\omega) \ = \ \sum_{\eta\in \mathbb{Z}^d} \e^{\im \vec{k}\cdot\frac{\xi-\eta}{2}} \left [ r(\xi,\eta) - r(0,\eta-\xi) \right ] f(\eta,\tau_{\frac{\eta-\xi}{2}} \omega).$
\end{enumerate}
\end{lem}
\begin{proof}
We define operators $T, V$ and $\mc{L}$ on $\mc{H}$ by
$$T \rho(X,\xi,\omega) \ := \ \sum_{|e|=1} \rho(X+e,\xi+e,\omega) - \rho(X+e,\xi - e,\omega)$$
$$
V \rho (X,\xi,\omega) \ := \ \left ( \omega\left ( \frac{X+\xi}{2} \right )-\omega\left ( \frac{X-\xi}{2} \right ) \right ) \rho (X,\xi,\omega)$$
and  
$$ \mc{L}\rho (X,\xi,\omega) \ : = \  \sum_{\eta\in \mathbb{Z}^d} \left [ r(\xi,\eta) - r(0,\eta-\xi) \right ] \rho(X,\eta,\omega)$$
so that the equation of motion for $\rho_{t}$ reads
$$\partial_t \rho_t = -\im u T \rho_t -\im \lambda V \rho_t + g \mc{L} \rho_t \ .$$
After Fourier transformation, this becomes
$$ \partial_t \wh{\rho}_{\vec{k};t} \ = \ -\im u \wh{T \rho}_{\vec{k};t} -\im \lambda \wh{V \rho}_{\vec{k};t} 
+ g \wh{\mc{L}\rho}_{\vec{k};t} \ .$$
Straightforward computations yield
\begin{align*}\wh{T \Psi}_{\vec{k}}(\xi,\omega)
\ &= \ \sum_{a} \e^{-\im \vec{k}\cdot a}  \sum_{|e|=1} \left [\Psi(\xi + 2a + e, \xi+e,\tau_a \omega) - \Psi(\xi+2a+e, \xi-e, \tau_a \omega)  \right ]\\
&= \ \sum_{|e|=1} \left [ \sum_{a} \e^{-\im \vec{k}\cdot a}  \Psi(\xi + 2a + e, \xi+e,\tau_a \omega) - 
\sum_{a} \e^{-\im \vec{k}\cdot a} \Psi(\xi-e+2(a+e), \xi-e, \tau_a \omega) \right ] \\
&= \
\sum_{|e|=1} \left [ \wh{\Psi}_{\vec{k}}(\xi+e,\omega) - \e^{\im \vec{k}\cdot e} \wh{\Psi}_{\vec{k}}(\xi-e,\tau_{-e}\omega) \right ] \\ &= \ \sum_{|e|=1} \left [ \wh{\Psi}_{\vec{k}}(\xi+e,\omega) - \e^{- \im \vec{k}\cdot e} \wh{\Psi}_{\vec{k}}(\xi+e,\tau_{e}\omega) \right ]
\end{align*}
and
\begin{align*}
\wh{V \Psi}_{\vec{k}} (\xi,\omega) \ &= \ \sum_{a} \e^{-\im \vec{k}\cdot a} \left ( \tau_a\omega(\xi +a ) - \tau_a\omega(a) \right ) 
\Psi(\xi + 2a, \xi, \tau_a \omega) \\
&=  \ \left (\omega(\xi) - \omega(0) \right ) \sum_{a} \e^{-\im \vec{k}\cdot a} \Psi(\xi + 2a, \xi, \tau_a \omega) \ = \ (\omega(\xi)-\omega(0)) \wh{\Psi}_{\vec{k}}(\xi,\omega) \ .
\end{align*}
Furthermore,
\begin{align*}
\wh{\mc{L} \Psi}_{\vec{k}} (\xi,\omega) \ &= \ \sum_{a} \e^{-\im \vec{k}\cdot a} \sum_{\eta} \left [ r(\xi,\eta) -r(0,\eta-\xi) \right ] \Psi(\xi +2a,\eta,\tau_a\omega) \\
&=  \ \sum_{\eta} \left [ r(\xi,\eta) -r(0,\eta-\xi) \right ]  \sum_{a}\e^{-\im \vec{k}\cdot a} \Psi(\xi +2a,\eta,\tau_a\omega) \\ &= \ \sum_{\eta}\e^{-\im \vec{k}\cdot \frac{\xi-\eta}{2}}  \left [ r(\xi,\eta) -r(0,\eta-\xi) \right ]  \wh{\Psi}_{\vec{k}}(\eta, \tau_{\frac{\eta-\xi}{2}} \omega) \ .\qedhere
\end{align*}
\end{proof}

We set 
\begin{equation}
\label{generator}
\mc{G}_k := \im \mc{A}_k - g \wh{\mc{L}}_{k} \ .
\end{equation} 
Lemma \ref{lem:main} has the following corollary.

\begin{lem}\label{lem:twoderivatives}
\begin{equation}\label{eq:twoderivatives}
\frac{1}{4} \sum_{X} X_1^2 \Ev{\rho_t(X,0,\omega)} \ = \ -\left . \partial_{k_1}^2 \ipc{\delta_0 \otimes 1}{ \e^{-t \mc{G}_{\vec{k}} } \delta_0 \otimes 1}_{\wh{\mc{H}}} \right |_{\vec{k}=0} \ .
\end{equation}
\end{lem}
\begin{rem*}Here
$$\delta_0\otimes 1(x,\omega) \ = \ \begin{cases} 1 & \text{ if } x=0 \ , \\
 0 & \text{ otherwise.}
\end{cases}$$
\end{rem*}
\begin{proof} 
First note that, by eq.\ \eqref{eq:FGV}, we are justified in interchanging differentiation and summation to write
$$\sum_{X} X_1^2 \Ev{\rho_t(X,0,\omega)}  \ = \ - \left . \partial_{k_1}^2 \sum_{X} \e^{-\im \vec{k} \cdot X} \Ev{\rho_t(X,0,\omega)} \right |_{\vec{k}=0} \ .$$
By shift invariance of the distribution of $\omega$, it  follows that
$$
\frac{1}{4} \sum_{X} X_1^2 \Ev{\rho_t(X,0,\omega)} \ = \ 
- \left . \partial_{k_1}^2 \sum_{x} \e^{-\im \vec{k}\cdot x} \Ev{\rho_t(2x,0,\tau_x \omega)} \right |_{\vec{k}=0} \
= - \left . \partial_{k_1}^2 \Ev{\wh{\rho}_{\vec{k},t}(0,\omega)} \right |_{k=0} \ .
$$
Since $\rho_0(X,\xi,\omega) = \delta_{0}(X)\delta_{0}(\xi)$, we have that $\wh{\rho}_{\vec{k},0}(\xi,\omega) = \delta_0(\xi)$.  Lemma \ref{lem:main} then implies that $\wh{\rho}_{\vec{k},t} \ = \ \e^{-t \mc{G}_k} \delta_0 \otimes 1$. This completes the proof.
\end{proof}
To compute the right hand side of eq.\ \eqref{eq:twoderivatives}, we begin by noting that
\begin{equation}
\begin{aligned} 
-&\partial_{k_1}^2   \left .  \ipc{\delta_0 \otimes 1}{ \e^{-t \mc{G}_{\vec{k}}} \delta_0 \otimes 1}_{\wh{\mc{H}}} \right |_{\vec{k}=0} &
\\ & = \ - 2 \int_0^t\di s \int_0^s  \di r \ipc{\delta_0\otimes 1}{ \e^{-(t-s)\mc{G}_{\vec{0}}} \left . \partial_{k_1} \mc{G}_{\vec{k}} \right |_{{\vec{k}}=0} \e^{- (s-r)\mc{G}_{\vec{0}}} \left . \partial_{k_1} \mc{G}_{\vec{k}} \right |_{{\vec{k}}=0} \e^{- r\mc{G}_{\vec{0}}} \delta_0\otimes 1 } \\
& \qquad +  \int_0^t \di s \ipc{\delta_0\otimes 1}{ \e^{-(t-s)\mc{G}_{\vec{0}}} \left . \partial_{k_1}^2 \mc{G}_{\vec{k}} \right |_{k=0} \e^{- s\mc{G}_{\vec{0}}} \delta_0\otimes 1} \\
& = \  - 2\int_0^t\di s \int_0^s  \di r \ipc{ \left . \partial_{k_1} \mc{G}_{\vec{k}}^\dagger \right |_{{\vec{k}}=0} \delta_0\otimes 1}{  \e^{- (s-r)\mc{G}_{\vec{0}}} \left . \partial_{k_1} \mc{G}_{\vec{k}} \right |_{{\vec{k}}=0} \delta_0\otimes 1 } \\
& \qquad +  \int_0^t \di s \ipc{\delta_0\otimes 1}{ \left . \partial_{k_1}^2 \mc{G}_{\vec{k}} \right |_{{\vec{k}}=0} \delta_0\otimes 1}
\end{aligned}\label{eq:acalc}
\end{equation}
where we have used that
\begin{equation} \e^{-t \mc{G}_{\vec{0}} } \delta_0 \otimes 1 \ = \ \e^{-t\mc{G}_{\vec{0}}^\dagger} \delta_0\otimes 1 \ = \ \delta_0\otimes 1\label{eq:k=0ev} \ .	
\end{equation}
Eq.\ \eqref{eq:k=0ev} follows from the following facts
\begin{enumerate}
\item $\mc{A}_{\vec{0}} \delta_0 \otimes 1 \ = \ 0$,
\item $\wh{\mc{L}}_{\vec{k}} \delta_0 \otimes 1 \ = \ 0$, \text{    }and
\item $ \wh{\mc{L}}_{\vec{k}}^\dagger \delta_0\otimes 1\  = \ 0$,
\end{enumerate}
which may be verified by explicit computation.  (In fact, (1) and (3) are quite general \tem \ they follow from the conservation of quantum probabilities.  However, (2) requires Item 1 of Assumption \ref{ass:r}.)

Since $\wh{\mc{L}}_{\vec{k}}\delta_0\otimes 1=\wh{\mc{L}}^{\dagger}_{\vec{k}}\delta_0\otimes 1 = 0$, for all 
$\vec{k}$, we find that
$$ \partial_{k_1} \mc{G}_{\vec{k}} \delta_0\otimes 1 \ = \ \im u \partial_{k_1} T_{\vec{k}} \delta_0\otimes 1 \ = \   u \left [ \e^{\im k_1} \delta_{-e_1}\otimes 1 - \e^{-\im k_1}\delta_{e_1}\otimes 1 \right ]   	
$$
$$ \partial_{k_1} \mc{G}_{\vec{k}}^\dagger  \delta_0\otimes 1 \ = \ - \im u \partial_{k_1} T_{\vec{k}} \delta_0\otimes 1  \ = \ -  u \left [ \e^{\im k_1} \delta_{-e_1}\otimes 1 - \e^{-\im k_1}\delta_{e_1}\otimes 1 \right ]$$
and 
$$\partial_{k_1}^2 \mc{G}_{\vec{k}} \delta_0\otimes 1 \  = \ \im u
\left [ \e^{\im k_1} \delta_{-e_1}\otimes 1 + \e^{-\im k_1}\delta_{e_1}\otimes 1 \right ]$$
where $e_1$ is the first basis vector.  Using these computations in eq.\ \eqref{eq:acalc}, we conclude that
$$ -\partial_{k_1}^2   \left .  \ipc{\delta_0 \otimes 1}{ \e^{-t \mc{G}_{\vec{k}}} \delta_0 \otimes 1}_{\wh{\mc{H}}} \right |_{{\vec{k}}=0} \ = \ 2 u^2 \int_0^t\di s \int_0^s  \di r \ipc{\phi}{   \e^{- (s-r)\mc{G}_0}  \phi } $$
where 
$$ \phi  \ =  \  (\delta_{e_1}-\delta_{-e_1})\otimes 1 \ .$$

Using Lemma \ref{lem:twoderivatives}, it now follows from the Tauberian theorems formulated in Ref.\ \onlinecite[Chapter XIII]{FellerII}, that
\begin{align*} \lim_{t\rightarrow \infty} \frac{1}{4t} \sum_{X} X_1^2 \Ev{\rho_t(X,0,\omega)} \ =& \ \lim_{t\rightarrow \infty} 
2 u^2 \frac{1}{t} \int_0^t\di s \int_0^s  \di r \ipc{\phi}{   \e^{- r \mc{G}_0}  \phi }\\
=& \ \lim_{\eta \rightarrow 0} 2 u^2 \eta \int_0^\infty \di t \, \e^{-t\eta} \int_0^t \di r\ipc{\phi}{   \e^{- r \mc{G}_0}  \phi } \\
=&  \ \lim_{\eta \rightarrow 0} 2 u^2 \int_0^\infty \di r \, \e^{-r \eta}
\ipc{\phi}{   \e^{- r \mc{G}_0}  \phi } \
= \ \lim_{\eta \rightarrow 0} 2u^2
\ipc{\phi}{   \frac{1}{\eta + \mc{G}_0}   \phi } \, 
\end{align*}
provided the limit on the right side exists. 

This limit can be shown to exist  using a straightforward ``Feshbach  argument.''  Recall that $\mc{G}_{\vec{0}}\delta_0 \otimes 1 = \mc{G}_{\vec{0}}^\dagger \delta_0 \otimes 1 = 0$.  It follows that $$\wh{\mc{H}}^\perp \ = \ \setb{f\in L^2(\Z^d\times \Omega)}{ \ipc{\delta_0\otimes 1}{f} =0}$$
is invariant under $\mc{G}_0$.  Since $\phi \in \wh{\mc{H}}^\perp$, we need only consider the restriction of $\mc{G}_0$ to $\wh{\mc{H}}^\perp$ to compute the resolvent matrix element of the resolvent  we wish to control. We further decompose the subspace $\wh{\mc{H}}^\perp$ as follows:
$$\wh{\mc{H}}^\perp \ = \ \wh{\mc{H}}^\perp_0 \oplus \wh{\mc{H}}^\perp_1$$
where
$$ \wh{\mc{H}}^\perp_0 \ = \ \setb{ f\in \wh{\mc{H}}^\perp }{f(x,\omega)=0 \text{ if $x\neq 0$.}} \ = \ \setb{\delta_0\otimes f}{\Ev{f(\omega)}=0}$$
and
\begin{equation}\label{subspace}
\wh{\mc{H}}^\perp_1 \ = \ \setb{f(x,\omega) \in \wh{\mc{H}} }{f(0,\omega) = 0} \ .
\end{equation} 
Let $Q_0$ denote the orthogonal projection of $\wh{\mc{H}}^\perp$ onto $\wh{\mc{H}}^\perp_0$ and set $Q_1:=1-Q_0$. We have the following block-matrix form for the action of $\mc{G}_0$ on  $\wh{\mc{H}}^\perp$:
$$\left . \mc{G}_0 \right |_{ \wh{\mc{H}}^\perp} \ \simeq \ \begin{pmatrix}
 0 & \im u Q_0 \wh{T}_0 Q_1 \\
 \im u Q_1 \wh{T}_0 Q_0 & Q_1 \mc{G}_0 Q_1
 \end{pmatrix} \ .$$
Since $\phi \in \widehat{\mc{H}}_{1}^{\perp}$, the Schur Complement formula yields
$$ \ipc{\phi}{   \frac{1}{\eta + \mc{G}_0}  \phi} \ = \ 
\ipc{\phi}{   \frac{1}{\eta + Q_1 \mc{G}_0 Q_1 + \frac{u^2}{\eta} Q_1 \wh{T}_0 Q_0\wh{T}_0 Q_1 }   \phi } \ .$$
Now
$$ \Re Q_1 \mc{G}_0 Q_1 \ = \ g  Q_1 \wh{\mc{L}} Q_1 \ \ge \ c g Q_1$$
where $c$ is the constant appearing in eq. \eqref{eq:exp-bound}. Thus

$$\norm{ \frac{1}{\eta + Q_1 \mc{G}_0 Q_1 + \frac{u^2}{\eta} Q_1 \wh{T}_0 Q_0\wh{T}_0 Q_1 } Q_1 } \ \le
 \ \frac{1}{c g} < \infty$$
uniformly in $\eta$.  

It follows that $\ipc{\phi}{   (\eta + \mc{G}_0 )^{-1} \phi}$ remains bounded as $\eta\rightarrow 0$; and it  remains to show that it has a non-zero limit.  To this end, let 
\begin{align*} \Pi \ : =& \ \text{ orthogonal projection of $\wh{\mc{H}}^\perp$ onto kernel of } Q_1 \wh{T}_0 Q_0\wh{T}_0 Q_1 \\ 
 =& \ \text{orthogonal projection  of $\wh{\mc{H}}^\perp$ onto kernel of }  Q_0\wh{T}_0 Q_1 \ .
 \end{align*}
Then we have (see Lemma \ref{lem:limres}, below)
$$\lim_{\eta \rightarrow 0} \frac{1}{\eta + Q_1 \mc{G}_0 Q_1 + \frac{u^2}{\eta} Q_1 \wh{T}_0 Q_0\wh{T}_0 Q_1} \ = \ \Pi \frac{1}{\Pi \mc{G}_0  \Pi} \Pi ,$$
where the limit is in the weak operator topology  and the operator $\Pi \mathcal{G}_{0} \Pi$ is boundedly-invertible on the range of $\Pi$, because
$$\Re \Pi \mc{G}_0 \Pi \ = \ 
\Pi \wh{\mc{L}} \Pi \ \ge  \ c g \Pi \ . $$ 

By an explicit computation, we see that given $f(x,\omega) \in \wh{\mc{H}}^\perp$,
$$  \left [ Q_0\wh{T}_0 Q_1 f\right ] (0,\omega) \ = \ \sum_{|e|=1} \left( f(e,\omega)-f(e,\tau_e\omega)\right) \ .$$
It follows that
$$  Q_0\wh{T}_0 Q_1 \left ( \delta_{e_1} - \delta_{-e_1} \right ) \otimes 1 = 0 \ .$$
Hence \text{   } $\phi = \left ( \delta_{e_1} - \delta_{-e_1} \right ) \otimes 1 \in \ran \Pi$, and
$$ \lim_{\eta \rightarrow 0} \ipc{\phi}{   \frac{1}{\eta + \mc{G}_0}  \phi} \ = \ \ipc{\phi}{\frac{1}{\Pi \mc{G}_0\Pi} \phi}\ .$$

Thus,
$$ D \ := \ \lim_{t\rightarrow \infty} \frac{1}{t} \sum_{x} x_1^2 \Ev{\dirac{x}{\rho_t}{x}} \ = \ 2 u^2 \ipc{\phi}{   \frac{1}{\Pi \mc{G}_0 \Pi }   \phi } \ .$$
Note that $D \le \frac{2u^2}{c g}$, and $ D \ge 0$, since it is a limit of positive quantities. Moreover, 
$$ D \ = \ \Re{D}  \ = \ 2 g u^2 \text{Re} \ipc{\frac{1}{\Pi \mc{G}_0 \Pi } \phi}{ \wh{\mc{L}}  \frac{1}{\Pi \mc{G}_0 \Pi }   \phi }
\ \ge \ 2 g u^2  c \norm{\frac{1}{\Pi \mc{G}_0 \Pi }   \phi }^2 \ \ge \ \frac{4 g u^2 c}{\norm{\mc{G}_0}^2} $$
since $\norm{\phi}^2=2$. This completes the proof of Theorem \ref{thm:main}.

\section{Diffusion without disorder \tem \ Theorem \ref{thm:dissbal}} In the last section, we have not made use of the disorder in our estimates.  Indeed, the result also holds when $\lambda =0$.  In fact, for $\lambda =0$, the calculations become much simpler, because we have that $\mathcal{A}_{0} = \widehat{T}_0 =0$, and hence
$$D \ = \ 2 u^2  \ipc{\phi}{\frac{1}{g \wh{\mc{L}}} \phi}\ . $$
By eq. \eqref{eq:exp-bound}, $ \langle \phi,\frac{1}{\mathcal{\widehat{L}}} \phi \ \rangle$ is bounded above by $\nicefrac{2}{c}$. Theorem \ref{thm:dissbal} follows.

\section{Perturbation Theory for $D$ in the large disorder regime \tem \ Theorem \ref{thm:dissloc}.}
Localization in the form of Eq. \eqref{eq:localization} implies that
$$D(0) \ = \  \lim_{t \rightarrow \infty } \frac{1}{t} \sum_{x}\abs{x_1}^2  \Ev{\abs{\dirac{x}{\e^{-\im H_\omega t} }{0}}^2} \ = \ 0\ .$$\\
In order to compare the diffusion constant at $g=0$ with the one at $g>0$, it is convenient to introduce a cut-off version of $D(g)$, namely
\begin{equation}\label{eq:Deta}
D(g,\eta) \ := \  \eta^2  \int_0^\infty \di t \, \e^{-\eta t} \sum_{x} \abs{x_1}^2 \Ev{\dirac{x}{\rho_t}{x}} 
\end{equation}
where $\rho_t$ satisfies eq. \eqref{eq:Lindblad}. 
The quantity $D(g,\eta)$ can be thought of as the diffusion constant on a time scale of $\nicefrac{1}{\eta}$. 
By Lemma \ref{lem:twoderivatives},
$$ D(g,\eta) \ := \   \left . \partial_{k_1}^2  \ipc{\delta_0\otimes 1}{\frac{\eta^2}{\mc{A}_k^{(0)} + g \wh{\mc{L}} + \eta} \delta_0 \otimes 1} \right |_{k=0}\ .$$
Following the calculations of the previous section, and using that
$$ u \phi \ = \ \partial_{k_1} \im \mc{A}_k |_{k=0}\delta_0 \otimes 1
\quad \text{and} \quad  \mc{A}_0\delta_0\otimes 1 \ =\ 0,$$ we see that 
\begin{equation} D(g,\eta) \ = \ 
2 u^2  \ipc{\phi}{\frac{1}{\im \mc{A}_0 + g \wh{\mc{L}}+ \eta} \phi } \label{eq:Deta2}
\end{equation}
and
$$ D(g) \ =  \ \lim_{\eta\rightarrow 0} D(g;\eta) \ .$$

For $g=0$, 
$$D(0,\eta) \ = \ \eta^2 \int_0^\infty \e^{-\eta t} \left [ \sum_{x} |x_1|^2  \Ev{\abs{\dirac{x}{\e^{-\im H_\omega t} }{0}}^2} \right ] \di t,$$
hence
$$ D(0,\eta) \ \le \ \eta \ell^2$$
with $\ell$ is as in eq.\ \eqref{eq:localization}.
As $D(0,\eta)$ is real and $\mc{A}_0$ is Hermitian, eq.\ \eqref{eq:Deta2} implies that
$$D(0,\eta) \ = \ \Re D(0,\eta) \ = \  2\eta u^2 \norm{\frac{1}{\im \mc{A}_0 + \eta} \phi}^2 \  .$$
Thus 
$$ \norm{\frac{1}{\im \mc{A}_0 + \eta} \phi}^2  \ = \  \frac{\eta}{ 2 u^2} \int_0^\infty \e^{-\eta t} \left [ \sum_{x} |x|^2  \Ev{\abs{\dirac{x}{\e^{-\im H_\omega t} }{0}}^2} \right ] \di t$$
and so
$$ \norm{\frac{1}{\im\mc{A}_0 + \eta} \phi}^2 \ \le \ \frac{\ell^2}{ 2u^2 }\ .$$
Because $\mc{A}_0$ is a Hermitian operator, 
$$ \norm{\frac{1}{\im\mc{A}_0 + \eta} \phi}^2 \ = \ \int_{\R} \abs{ \frac{1}{i t + \eta}}^2 \di \mu_\phi(t)$$
where $\mu_\phi$ is the spectral measure of $\mc{A}_0$ associated to $\phi$. By monotone convergence,
$$ \int_{\R} \frac{1}{t^2} \di \mu_\phi(t) \ = \ \lim_{\eta \rightarrow 0} \norm{\frac{1}{\im\mc{A}_0 + \eta} \phi}^2 \ \le  \ \frac{\ell^2}{2u^2 d} \ .$$
Thus, \text{   }$\psi  =  \lim_{\eta\rightarrow 0} \frac{1}{\im\mc{A}_0 + \eta} \phi$ \text{    }exists in $\wh{\mc{H}}$, and   
$$ \norm{\psi }^2 \ \le \ \frac{\ell^2}{2u^2}$$
Note that 
$$\overline{\psi} \ = \ \lim_{\eta\rightarrow 0} \frac{1}{-\im \mc{A}_0+ \eta} \phi$$
where $\overline{\psi}$ is the complex conjugate of $\psi$.  (This is true, because $\phi$ is a real function and 
$\mc{A}_0$ is a \textit{real} symmetric operator.)

Translated back to the Schr\"odinger operator picture, we have proven the existence of the limit
$$\ell_0^2 \ := \ \lim_{\eta \rightarrow 0} \eta \int_0^\infty \e^{-\eta t} \left [ \sum_{x} \frac{|x|^2}{d} \Ev{ \abs{\dirac{x}{\e^{-\im t H_\omega}}{0}}^2} \right ] \di t$$
and obviously $\ell_0 \le \ell$.  Note that $ \norm{\psi }^2 \ = \ \norm{\overline{\psi}}^2 \ = \  \frac{\ell_0^2}{2u^2 d}.$

Returning to $g>0$, we have that
\begin{multline*} 
D(g,\eta) \ = \ D(0,\eta) - 2u^2  g \ipc{\phi}{\frac{1}{\im\mc{A}_0 + \eta}\wh{\mc{L}} \frac{1}{\im\mc{A}_0 + \eta} \phi} \\ + 2u^2  g^2 \ipc{\phi}{\frac{1}{\im\mc{A}_0 + \eta}\wh{\mc{L}} \frac{1}{\im\mc{A}_0  + g\wh{\mc{L}} + \eta}\wh{\mc{L}} \frac{1}{\im\mc{A}_0 + \eta} \phi} \ .
\end{multline*}
Since $\wh{\mc{L}} = Q_1 \wh{\mc{L}} Q_1$, where $Q_1$ is the orthogonal projection onto the subspace 
$\widehat{\mathcal{H}}_{1}^{\perp}$ introduced in \eqref{subspace}, we have that
$$
\wh{\mc{L}} \frac{1}{\im\mc{A}_0  + g\wh{\mc{L}} + \eta}\wh{\mc{L}}
\ = \ \wh{\mc{L}}\frac{1}{\im Q_1\mc{A}_0Q_1 + g \wh{\mc{L}} + \eta + \frac{1}{\eta} Q_1\wh{T}_0 Q_0 \wh{T}_0 Q_1} \wh{\mc{L}} \ .
$$
Taking the limit $\eta \rightarrow 0$, as in the proof of Theorem \ref{thm:main}, we find that
$$D(g) \ = \ - 2u^2 g \ipc{\overline{\psi}}{\left [ \wh{\mc{L}} - g  \wh{\mc{L}} \Pi \frac{1}{\im \Pi\mc{A}_0 \Pi + g \Pi\wh{\mc{L}}\Pi} \Pi \wh{\mc{L}}  \right ] \psi} $$
where $\Pi$ is the projection onto the kernel of $Q_0\hat{T}_0 Q_1$, as above. In particular, we obtain the upper bound
$$\abs{D(g)} \ \le \ 2 u^2 g \left ( 1 + \frac{1}{c} \right ) \norm{\psi}^2 
\ = \ \left ( 1 + \frac{1}{c} \right ) \ell_0^2 g$$
where we have used that $\| \wh{\mc{L}}\|\leq1$ (by convention).
Furthermore, we have the lower bound
$$D(g)  \ \ge \  \frac{4u^{2}gc}
{\Vert \mathcal{G}_{0}\Vert ^{2}} \ = \ \frac{4u^2 gc}{\Vert\mathcal{A}_0 + g \wh{\mc{L}}\Vert^2}$$
obtained at the end of the proof of Theorem \ref{thm:main}.
We conclude that
\begin{equation}
\label{dep. on g}
D(g) \ = \ \mc{O}(g), \qquad \text{as} \text{     } g\rightarrow 0 \ . 
\end{equation}
In fact, we can take the limit 
\begin{equation}\label{eq:limit} \Delta  \ := \ \lim_{g \downarrow 0} \frac{D(g)}{g} \ .
\end{equation}
Let $\Pi_0$ denote the projection onto the kernel of $\Pi \mc{A}_0 \Pi$.  If $\Pi_0=0$ then 
$$ \frac{g}{\im \Pi \mc{A}_0  \Pi + g \Pi \wh{\mc{L}}\Pi }  \ \xrightarrow[]{g\rightarrow 0} \ 0 $$
in the weak operator topology; see Lemma \ref{lem:limres}.  It seems likely that $\Pi_0=0$, but we are not aware of a proof.  Furthermore we do not need to resolve this issue, since if $\Pi_0 \neq 0$ then
$$ \frac{g}{\im \Pi \mc{A}_0  \Pi + g \Pi \wh{\mc{L}}\Pi }  \ \xrightarrow[]{g\rightarrow 0} \  \Pi_0 \frac{1}{\Pi_0 \wh{\mc{L}} \Pi_0} \Pi_0 $$
in the weak operator topology; see Lemma \ref{lem:limres}. In any case, the limit in eq.\ \eqref{eq:limit} exists, and
$$ \Delta  \ = \  - 2 u^2  \ipc{\overline{\psi}}{\left [ \wh{\mc{L}} -   \wh{\mc{L}} \Pi_0 \frac{1}{\Pi_0 \wh{\mc{L}}\Pi_0} \Pi_0 \wh{\mc{L}}  \right ] \psi}
$$
where  we take $\Pi_0  (\Pi_0 \wh{\mc{L}}\Pi_0)^{-1} \Pi_0 =0$ \text{  } if $\Pi_0=0$.  Note that we have the estimates
$$ \frac{4 c}{\|\mc{A}_0\|^2} u^2 \ \le \  \Delta  \ \le \  \left(1+\frac{1}{c}\right)\ell_0^2 \ .$$\\
\begin{acknowledgements}
	We thank the School of Mathematics of the Institute for Advanced Study and, in particular, Thomas C. Spencer for generous hospitality during a period when we started to collaborate on the problems solved in this paper and generated most of the ideas used in our analysis. The stay of J.F. at the Institute for Advanced Study was supported by `The Fund for Math' and `The Robert and Luisa Fernholz Visiting Professorship Fund.'  J.S. was supported by United States NSF grants DMS-0844632 and DMS-1500386, and his stay at the Institute for Advanced Study was supported by `The Fund for Math.' This manuscript was completed during a stay of J.S. at the Isaac Newton Institute for the program `Periodic and Ergodic Spectral Problems.'
\end{acknowledgements}

\appendix
\section{Finite group velocity estimates for the Lindblad equation}
\begin{lem}\label{lem:FGVE}
Let $\rho_0$ be a given density matrix such that $$A \ = \ \sup_{X}\e^{m|X|}  \left ( \sum_{\xi} \abs{\rho_0(X,\xi)}^2 \right )^{\frac{1}{2}} \ < \ \infty \ .$$ If, for $t>0$, the density matrices $\rho_t$ satisfy eq.\ \eqref{eq:Lindblad}, with initial condition $\rho_0$, then 
$$ \sup_{X} \e^{m|X|} \left (  \sum_{\xi} \abs{\rho_t(X,\xi)}^2 \right)^{\frac{1}{2}} \ \le \ \e^{C_m  t} A \ , $$
with $C_m = 4d \e^m u + g $.
\end{lem}

\begin{proof} Let $\mc{B}_m$ denote the Banach space of functions $F:\Z^d \times \Z^d \rightarrow \bb{C}$ such that
$$\norm{F}_m \ := \ \sup_{X} \left (  \sum_{\xi} \e^{m|X|} \abs{F(X,\xi)}^2 \right )^{\frac{1}{2}} \ < \ \infty\ .  $$
Thus $\rho_0\in \mc{B}_m$ by assumption.  Let $\phi(X,\xi) = \im \lambda \left ( \omega\left ( \frac{X+\xi}{2}\right ) - \omega\left ( \frac{X-\xi}{2}\right ) \right )$ and let 
$$W_t(X,\xi) \ = \ \e^{\im t \phi(X,\xi)} \rho_t(X,\xi)\ .$$ 
Note that $\abs{W_t(X,\xi)}=\abs{\rho_t(X,\xi)}$. In particular, $W_0 \in \mc{B}_m$. 

The evolution of $W_t$ is governed by the non-autonomous equation
$$\partial_t W_t \ = \  \mc{G}_t W_t\ ,$$
with time-dependent generator
$$ \mc{G}_t \ = \ \e^{\im t \phi} (\im u T + g \mc{L} ) \e^{-\im t \phi}\ ,$$
where the ``kinetic hopping operator" $T$ is
$$ T\rho(X,\xi) \ = \ \sum_{|e|=1}\left [ \rho(X+e,\xi+e)-\rho(X+e,\xi-e)\right ]\ .$$
Because $\phi$ and $\mc{L}$  fiber over $X$, one easily computes that
\begin{multline*}\norm{\e^{\im t \phi} \mc{L} \e^{-\im t \phi} F}_m \\ = \ \sup_{X} \e^{m|X|} \left ( \sum_{\xi} \abs{\sum_{\eta} \left ( r(\xi,\eta) - r(0,\eta-\xi) \right ) \e^{-\im t\phi(X,\eta)} F(X,\eta)}^2 \right )^{\frac{1}{2}}
\ \le \	 \norm{F}_m \ ,
\end{multline*}
by the normalization eq.\ \eqref{eq:normalization}.
On the other hand, the hopping term $\e^{\im t\phi} T\e^{-\im t \phi}$ is bounded by
\begin{multline*}\norm{\e^{\im t \phi} T \e^{-\im t \phi} F}_m \\ =  \ \sup_{X} \e^{m|X|}  \left ( \sum_{\xi} \abs{ \sum_{|e|=1} \e^{\im t\phi(X+e,\xi-e)} F(X+e,\xi+e) - \e^{\im t\phi(X+e,\xi-e)} F(X+e,\xi-e)}^2 \right )^{\frac{1}{2}} \\
\le \ 2 \sup_{X} \e^{m|X|} \sum_{|e|=1} \left ( \sum_{\xi} \abs{F(X+e,\xi)}^2 \right )^{\frac{1}{2}} \\ \le \ 2 \left [\sup_{X} \sum_{|e|=1} \e^{m (|X|-|X+e|)} \right ] \norm{F}_m \
\le \ 4 d \e^{m}  \norm{F}_m \ . \end{multline*}
Thus
$$\norm{\mc{G}_t F}_m \ \le \ 4 d \e^{m} u +  g $$
and the result follows.
\end{proof}

\section{A limiting principle for resolvents}

\begin{lem}\label{lem:limres} Let $\mc{H}$ be a Hilbert space. Let $A$ be a normal operator on $\mc{H}$ and $B$ a bounded operator on $\mc{H}$, with $\Re A \ge 0$ and $\Re B \ge c >0$.  
\begin{enumerate}
\item If $\ker A = \{0\}$, then 
$$\lim_{\lambda \rightarrow \infty} \ipc{\phi}{\left ( \lambda A + B \right )^{-1} \psi}_{\mc{H}} \ = \ 0 $$
for any $\phi,\psi \in \mc{H}$.
\item If $\ker A \neq \{ 0 \}$, then  
$$\lim_{\lambda \rightarrow \infty} \ipc{\phi}{\left ( \lambda A + B \right )^{-1} \psi}_{\mc{H}} \ = \ \ipc{\Pi \phi}{ \left ( \Pi B \Pi \right )^{-1} \Pi \psi}_{\ran \Pi}$$
for any $\phi,\psi\in \mc{H}$, where $\Pi=$ projection onto the kernel of $A$.
\end{enumerate}
\end{lem}
\begin{proof}
Let $\psi \in \mc{H}$ be given and let $h_\lambda = (\lambda A + B)^{-1}\psi$.   Note that  $\|h_\lambda\| \ \le \ c^{-1} \norm{\psi}$, since
$$ c \norm{h_\lambda}^2 \ \le \ \Re \ipc{h_\lambda}{\psi} \ \le \ \norm{h_\lambda} \norm{\psi}.$$
This and the identity $\lambda  A h_\lambda = \psi - B h_\lambda $ imply 
$$|\lambda| \norm{A h_\lambda} \ \le \ \left ( 1 + c^{-1} \norm{B} \right ) \norm{\psi} \ .$$
Thus  $(I-\Pi) h_\lambda$ converges weakly to zero.  If $\ker A = \{0\}$, so $\Pi=0$, then this completes the proof.   On the other hand, if $\Pi \neq 0$, then it commutes with $A$, since $A$ is normal. Thus
$$\Pi B h_\lambda \ = \ \Pi \psi \ .$$
Since $(I-\Pi) h_\lambda$ converges weakly to  $0$ and $\Pi B \Pi$ is boundedly invertible on $\ran \Pi$, it follows that $\Pi h_\lambda$ converges weakly to  $(\Pi B \Pi)^{-1} \Pi \psi$.
\end{proof}


\end{document}